\documentclass[11pt]{article}
\usepackage{fullpage}
\usepackage[utf8]{inputenc}
\usepackage{algorithm}
\usepackage{algorithmic}
\usepackage{amsfonts}
\usepackage{graphicx}
\usepackage{paralist}
\usepackage{amsthm}
\usepackage{mathrsfs}
\usepackage{amssymb}
\usepackage{hyperref}
\usepackage{amsmath}
\usepackage{cite}

\newcommand{\qedsymb}{\qed}
\newenvironment{proofof}[1]{\begin{trivlist}
		\item[\hspace{\labelsep}{\bf\noindent Proof of #1: }]
	}{\qedsymb\end{trivlist}}

\usepackage{xcolor}
\newtheorem{theorem}{Theorem}
\newtheorem{lemma}{Lemma}
\newtheorem{definition}{Definition}
\newtheorem{claim}{Claim}
\newtheorem{corollary}{Corollary}
\newtheorem{observation}{Observation}

\newtheorem{oq}{Open Question}

\usepackage{xspace}
\newcommand{\clique}{\textsc{Congested-Clique}\xspace}
\newcommand{\cliquebcast}{\textsc{Clique-Broadcast}\xspace}
\newcommand{\congestSTD}{\textsc{Congest}\xspace}
\newcommand{\congest}[1]{\textsc{Congest\ensuremath{(#1)}}\xspace}
\newcommand{\local}{\textsc{Local}\xspace}

\newcommand\floor[1]{\lfloor#1\rfloor}

\newcommand{\set}[1]{\left\{ #1 \right\}}
\DeclareMathOperator{\false}{{\scriptstyle{FALSE}}}
\DeclareMathOperator{\true}{{\scriptstyle{TRUE}}}

\newenvironment{theorem-repeat}[1]{\begin{trivlist}
		\item[\hspace{\labelsep}{\bf\noindent Theorem \ref{#1} }]\em }%
	{\end{trivlist}}

\newenvironment{lemma-repeat}[1]{\begin{trivlist}
		\item[\hspace{\labelsep}{\bf\noindent Lemma \ref{#1} }]\em }%
	{\end{trivlist}}

\begin{document}
	\begin{titlepage}
		\title{Fooling Views: A New Lower Bound Technique for Distributed Computations under Congestion}
		\author{Amir Abboud\footnote{IBM Almaden Research Center. \texttt{amir.abboud@ibm.com}.} \and Keren Censor-Hillel\footnote{Technion, Department of Computer Science, \texttt{\{ckeren,serikhoury\}@cs.technion.ac.il}. Supported in part by the Israel Science Foundation (grant 1696/14).}
			\and Seri Khoury\footnotemark[2]{} \and Christoph Lenzen\footnote{MPI for Informatics, Saarland Informatics Campus, \texttt{clenzen@mpi-inf.mpg.de}.}
		}
		\maketitle
		\begin{abstract}
			We introduce a novel lower bound technique for distributed graph algorithms under bandwidth limitations.
			We define the notion of \emph{fooling views} and exemplify its strength by proving two new lower bounds for triangle membership in the \congest{B} model:
			\begin{enumerate}
				\item Any $1$-round algorithm requires $B\geq c\Delta \log n$ for a constant $c>0$.
				\item If $B=1$, even in constant-degree graphs any algorithm must take $\Omega(\log^* n)$ rounds.
			\end{enumerate}
			The implication of the former is the first proven separation between the \local and the \congestSTD models for deterministic triangle membership.
			The latter result is the first non-trivial lower bound on the number of rounds required, even for \emph{triangle detection}, under limited bandwidth.
			All previous known techniques are provably incapable of giving these bounds.
			We hope that our approach may pave the way for proving lower bounds for additional problems in various settings of distributed computing for which previous techniques do not suffice.
		\end{abstract}
		\thispagestyle{empty}
	\end{titlepage}
	
	\section{Introduction}	
	
	In a group of $n$ players with names of $O(\log n)$ bits, how many days does it take for someone to find out if there are three players that are all friends with each other?
	All computation is free, and in the beginning, all players know their friends. Then, on each day, each player can send $B=O(\log{n})$ bits privately to each of its friends.
	
	This multi-party communication puzzle is known as the triangle detection problem in the popular \congestSTD model of distributed computing. Its complexity is poorly understood.
	In the naive protocol, each player (node) tells all its friends about all its other friends (sends its neighborhood to every neighbor).
	This takes a single round in the \local model and $O(n)$ rounds in \congestSTD.
	A clever randomized protocol of Izumi and Le Gall~\cite{IzumiG17} provides a solution with $O(n^{2/3}(\log{n})^{2/3})$ rounds in \congestSTD, and this is essentially all we know about the problem.
	Before our work, it could not be ruled out that the problem can be solved in $O(1)$ rounds, or even a single round, even with a bandwidth of $B=1$!
	
	\begin{oq}
		\label{oq1}
		What is the round complexity of triangle detection in the \congestSTD model of distributed computing?
	\end{oq}
	
	Triangle detection is an extensively studied problem in most models of computation.
	In the centralized setting, the best known algorithm involves taking the cube of the adjacency matrix of the graph. It runs in $O(n^{2.3729})$ time and was found using a complex computer program \cite{Vass12,LeGall14fmm}.
	If one wishes to avoid the impractical matrix multiplication, the problem can be solved in $O(n^3/\log^{4} n)$ time \cite{BW12,Chan15,Yu15}.
	Other works have designed algorithms for sparse graphs \cite{AYZ95}, for real-world graphs (e.g. \cite{SW05}), for listing all the triangles \cite{IR78,BPVZ14}, for approximately counting their number \cite{KMPT12}, for weighted variants (e.g. \cite{CL09}), and much more (an exhaustive list is infeasible).
	Moreover, conjectures about the time complexity of triangle detection and of its variants \cite{VW10,AV14,AVY15,HenzingerKNS15,ABV15b} are among the cornerstones of fine-grained complexity (see \cite{Vass15}).
	Other highly non-trivial algorithms were designed for it in settings such as: distributed models \cite{DLP12,CKK+15,DruckerKO13,KMRS15,CFSV16,FRST16,IzumiG17}, quantum computing \cite{Buhrman01,Szegedy03,MSS07,Belvos12,LMS13,LeGall14,LN17a,LN17b}, and Map-Reduce \cite{SASU13,AFU13}.
	It is truly remarkable that such a basic problem has lead to so much research.
	
	From a technical perspective, Open Question~\ref{oq1} is one of the best illustrations of the lack and necessity of new techniques for proving lower bounds in distributed computing.
	In this work, we present a novel lower bound technique, providing a separation between the \local and \congest{B} models for a problem for which previous techniques are \emph{provably incapable} of doing so.\footnote{\congest{B} stands for the synchronous model with a bandwidth of $B$ bits. In particular, the standard \local and \congestSTD models correspond to \congest{\infty} and \congest{\log{n}}, respectively.} This is the first progress towards answering Open Question~\ref{oq1} from the lower bounds side.
	
	\subsection{Prior lower bound techniques and their limits}
	
	To date, there are essentially two techniques for deriving lower bounds for distributed graph algorithms.
	The first is the \emph{indistinguishability} technique of Linial~\cite{Linial92}, which is the main source for lower bounds in the \local model, where the message size is unrestricted.
	This technique argues that any $r$-round algorithm, regardless of message size, can be seen as a function that maps the $r$-hop neighborhood of a node to its output.
	Here, the topology of the graph is labeled by the unique $O(\log n)$-bit node identifiers and any other input provided to the nodes;
	for randomized algorithms, we simply give each node an infinite string of unbiased random bits as part of the input.
	
	This technique has resulted in a large number of \emph{locality} lower bounds, e.g.~\cite{Brandt16,KuhnMW16,Linial87,Linial92,Naor91}.
	For these problems, it is a long-standing open question whether higher lower bounds can be found in the \congest{B} model, see e.g.~\cite{Pai+17}.
	Note that part of its appeal is that one can entirely ``forget'' about the algorithm:
	for instance, a $2$-round coloring algorithm is just interpreted as a function assigning a color to each possible $2$-neighborhood, and a correct algorithm must assign distinct colors to any pair of neighborhoods that may belong to adjacent nodes in \emph{any} feasible input graph;
	more generally, this gives rise to the so-called \emph{$r$-neighborhood graph,} and showing that $r$ rounds are insufficient for coloring with $c$ colors equates to showing that the chromatic number of the $r$-neighborhood graph is larger than~$c$.
	
	Unfortunately, as this technique does not take bandwidth restrictions into account, it cannot show any separation between the \local and \congest{B} models.
	Triangle detection is possibly one of the most extreme examples for this, as it can be solved in a single round in \local but seems to require $n^{\Omega(1)}$ rounds in \congest{\log n}.
	Additional examples are symmetry breaking problems, such as Maximal Independent Set (MIS) and $(\Delta+1)$-Coloring, upon which we elaborate in Section~\ref{sec:discussion}.

	The second tool available for generating distributed lower bounds is the \emph{information bottleneck} technique, first introduced implicitly by Peleg and Rubinovich \cite{PelegR00}.
	Here, the idea is to reduce a $2$-party communication complexity problem (typically set disjointness) to a distributed problem, and argue that a fast distributed algorithm under limited bandwidth would imply a protocol exchanging few bits.
	This approach yields a large number of strong lower bounds for a wide range of \emph{global} problems, in which there is no bound on the distance up to which a local change in the input may affect the output. Examples are, e.g.,~\cite{SarmaHKKNPPW12,FrischknechtHW12,Censor-HillelKPY16,NanongkaiSP11,LenzenP14}, but a complete list would justify an entire survey by its own.

	Lower bounds based on information bottlenecks can also be proven for local problems.
	For instance, Drucker et al.\ show an $n^{\Omega(1)}$ lower bound for detecting $k$-cliques or $k$-cycles for any fixed $k>3$~\cite{DruckerKO13}.
	However, this technique is inherently incapable of proving lower bounds for many problems.
	In particular, it completely fails for the problem of triangle detection.
	As discussed by Drucker et al., this is because no matter how we divide the nodes of the graph among the two players, one of them will know about the triangle.
	One may, in principle, hope for lower bounds based on multi-party communication complexity or information complexity, but to date no such result is known.
	
	Intuitively, a lower bound for triangle detection must combine the two techniques. We have to argue that when a small number of bits is sent, the nodes do not have enough information to distinguish between a distance-2 neighborhood in which there is a triangle and one in which there isn't. Interestingly, Drucker et al.\ prove that no such technique is possible without breakthroughs in circuit complexity in the related \clique model~\cite{DruckerKO13}, where nodes can send messages to all other nodes, not only neighbors.\footnote{This result applies to small output, e.g.\ the triangle detection problem, which is a decision problem. For \emph{listing} all triangles in the graph, a tight bound of $\Omega(n^{1/3}/\log n)$ holds~\cite{IzumiG17,DLP12}.} Indeed, it is still open whether triangle detection can be solved in $n^{o(1)}$ rounds even in this powerful model (it is likely that it is \cite{DruckerKO13}), but now we know that the lack of super-constant lower bounds can be blamed on our inability to prove ``computational hardness'' results in CS: in a similar vein, we do not know whether 3SAT can be solved in polynomial or even linear time.
	In contrast, no such barrier is known in the standard \congest{B} model, where communication is limited to the input graph.
	This is an embarrassing situation, as we have a huge gap, but no well-known barrier to blame it on.
	
	With a clever usage of Rusza-Szemerdi graphs, Drucker et al.\ were also able to prove a strong $n^{1-o(1)}$ lower bound for triangle detection~\cite{DruckerKO13}, under the restriction that each node sends the same message to all of its neighbors in each round (in a broadcast fashion). More specifically, they show that a deterministic protocol in the \cliquebcast model (and therefore also Congest-Broadcast) requires $\Omega(n/e^{\sqrt{\log{n}}})$ rounds, and that (essentially) the same lower bound holds for randomized protocols under the Strong Exponential Time Hypothesis, a popular conjecture about the time complexity of $k$-SAT.
	Unfortunately, even under such conjectures, we do not know how to get any non-trivial lower bound in the standard \congest{B} model.
	
	Finally, it is worth pointing out a subtlety about the statement that $2$-player communication complexity cannot provide \emph{any} lower bound for triangle detection.
	This statement is fully accurate only under the assumption that nodes initially know the identifiers of their neighbors, as this renders it trivial to infer from the joint view of two neighbors whether they participate in a triangle.
	This assumption is known as $KT_1$, where $KT_i$ means \emph{Knowledge of Topology} up to distance $i$ (excluding edges with both endpoints in distance $i$), as was first defined in~\cite{AwerbuchGPV90}.
	The difference between $KT_0$, in which a node knows only its own identifier, and $KT_1$, in which a node knows also the identifiers of its neighbors, has been a focus of abundant studies, in particular concerning the \emph{message complexity} of distributed algorithms (see, e.g.,~\cite{AwerbuchGPV90,PaiPPR17,KingKT15,Elkin17,Pandurangan17} and references therein).
	
	Note that acquiring knowledge on the neighbors' identifiers requires no more than sending $O(\log n)$ bits over each edge, so the distinction between $KT_0$ and $KT_1$ is insubstantial for the round complexity in \congest{B} for $B\in \Omega(\log n)$;
	$KT_1$ is therefore the default assumption throughout wide parts of the literature.
	However, in $KT_0$ a lower bound of $\Omega(\tfrac{\log n}{B})$ on the round complexity of triangle detection follows from a simple counting argument.
	As, ultimately, the goal is to show lower bounds of $\omega(\log n)$, we consider $KT_1$ in this work.
	
	\subsection{Our contribution}
	
	In this paper, we introduce \emph{fooling views}, a technique for proving lower bounds for distributed algorithms with congestion.
	We are able to show the first non-trivial round complexity lower bounds on triangle detection in $KT_1$, separating the \local and \congest{B} models:
	\begin{enumerate}
		\item Triangle membership\footnote{In the triangle membership problem, every node must indicate whether it participates in a triangle.} in one round requires $B\geq c\Delta \log n$ for a constant $c>0$ (Section~\ref{sec:bandwidth}).
		\item If $B=1$, triangle detection requires $\Omega(\log^* n)$ rounds, even if $\Delta=2$, and even if the size of the network is constant and $n$ is the size of the namespace (Section~\ref{sec:OneBit}).
	\end{enumerate}
	We stress that we do not view our main contribution as the bounds themselves:
	while the bandwidth lower bound for single-round algorithms is tight, it hardly comes as a surprise that such algorithms need to communicate the entire neighborhood. Additionally, we do not believe that with $1$-bit messages, extremely fast triangle detection is possible.
	Rather, we present a novel \emph{technique} that enables to separate the two models, which is infeasible with prior lower bound techniques.
	We hope this to be a crucial step towards resolving the large gap between lower and upper bounds, which in contrast to other models is not justified by, e.g., conditional hardness results.

	The basic idea of fooling views is that they combine reasoning about locality with bandwidth restrictions.
	Framing this in terms of neighborhood graphs, this would mean to label neighborhoods by the information nodes have initially and the communication the algorithm performs over the edges incident to a node.
	However, this communication depends on the algorithm and the communication received in earlier rounds, enforcing more challenging inductive reasoning to prove multi-round lower bounds.
	
	To capture the intuition for our technique for $B=1$, think about a node that receives the same messages from its neighbors regardless of whether it participates in a triangle or not as a \emph{fooled node}.
	Intuitively, in a triangle $\{u,w,v\}$ of a given network, if one of the nodes $u,v$ and $w$ is able to detect the triangle after $t$ rounds of communication, then it may simply inform the other two nodes about the triangle during round number $t+1$.
	Thus, it is crucial to maintain a perpetual state of confusion for all nodes involved in the triangle.
	
	However, if the task is to detect whether a specific triple of IDs is connected by a triangle or not, then the nodes can solve this by simply exchanging only one bit of communication.
	Accordingly, our goal is to keep a large subset of the \emph{namespace} fooled as long as possible.
	To this end, think of a triangle $\{u,v,w\}$ for which none of $u,v$ and $w$ is able to detect the triangle as a \emph{fooled triangle}.
	Our main idea is to show that if there are many fooled triangles after $t$ rounds, then there are many triangles among them that are fooled after $t+1$ rounds as well.
	In order to express this intuition, one of our ingredients in the proof is the following extremal combinatorics result by Paul Erd\"os~\cite{Erdos1964}.
	\begin{theorem}[\cite{Erdos1964}, Theorem 1]
		\label{thm: large hypergraphs1}
		Any k-uniform hypergraph of n nodes which contains at least $n^{k-\ell^{1-k}}$ edges, must contain a complete $k$-partite $k$-uniform hypergraph such that each part of it is of size $\ell$.
	\end{theorem}
	Using this theorem of Erd\"os, we are able to show that if there are many fooled triangles after $t$ rounds, then there is a set of nodes such that each triple in the set is a fooled triangle after $t+1$ rounds.
	Blending counting and indistinguishability arguments with this theorem, we can derive our lower bound for multi-round algorithms.
	
	Our $\Omega(\log^* n)$ bound serves as a proof of concept that our technique has the power to break through the bounds of previous techniques by demonstrating that this is indeed possible.
	Note that purely information-theoretic reasoning runs into the obstacle that in the $KT_1$ model, $\Theta(\log n)$ bits have already crossed each edge ``before the algorithm starts.''
	Accordingly, we argue that our approach represents a qualitative improvement over existing techniques.
	Proving lower bounds higher than $\log^*{n}$ requires new ideas, but we are hopeful that combining our technique with a more sophisticated analysis will lead to much higher lower bounds, both for triangle detection and for other non-global problems discussed above.
	
	As an additional indication that the proposed technique is of wider applicability, we apply it to $k$-cycle detection for $k>3$, for showing lower bounds on the bandwidth of optimal-round algorithms.
	However, here the information bottleneck technique is applicable again, and we do not obtain stronger bounds using our fooling views.
	The details of these constructions are given in Appendix~\ref{sec:kcycle-reduction};
	the main body of the paper focuses on triangle detection. We conclude with open questions in Section~\ref{sec:discussion}.
	
	\subsection{Further related work}\label{sec:related}
	
	\paragraph{Edge-crossings.}
	The basic topology components for our lower bound in Section~\ref{sec:OneBit} are triangles and $6$-cycles. The main hardness that we show for a node in deciding whether it participates in a triangle or not comes from not knowing the neighbors of its neighbor. That is, a node is unable to distinguish between two triangles and one $6$-cycle, because only difference between these two cases is a single \emph{edge crossing}, which are two node-disjoint edges for which we swap the endpoints from $\{w,x\},\{w',x'\}$ to $\{w,x'\},\{w',x\}$. Edge crossings have previously aided the construction of lower bounds, such as lower bounds for message complexity of broadcast~\cite{AwerbuchGPV90} or of symmetry breaking~\cite{PaiPPR17}, as well as lower bounds for proof-labeling schemes~\cite{MorFP15}.
	
	\paragraph{The \congest{1} model.}
	While we view our results for the \congest{1} model more as a proof of concept for our technique rather than as a bound that attempts to capture the true complexity, this model has been attracting interest by itself in previous work. It has been shown in~\cite{MetivierRSZ11} that the cornerstone $O(\log n)$-round algorithms for maximal independent set~\cite{Luby86,AlonBI86} can be made to work with even with a bandwidth of $B=1$, and this problem was studied also in, e.g.,~\cite{Bar-NoyNN90,KothapalliSOS06}. Using the standard framework of reduction from $2$-party communication complexity, a $\Omega(\sqrt{n})$ lower bound for the number of rounds required for $4$-cycle detection can be directly deduced from~\cite{DruckerKO13}. In fact, all lower bounds obtained using this framework are with respect to the bandwidth $B$, and therefore imply lower bounds also for the case of $B=1$.

	\section{Model and Definitions}
	Our model is a network of $n$ nodes, each having an ID in $[N]$, for some polynomial $N$ in $n$. Each node starts with knowledge of its ID as well as the IDs of its neighbors. This is known as the $KT_1$ model, and differs from the $KT_0$ model in which each node starts only with knowledge of its own ID. The nodes communicate in synchronous rounds, in which each node can send a $B$-bit messages to each of its neighbors.
	
	The model we consider is the \congest{B} model~\cite{Peleg00}, where $B$ is the bandwidth that is given for each message\footnote{For simplifying the exposition, we will assume that nodes send $B$ bits in each round and cannot send less bits or remain silent. It is easy to verify that this only affects the constants in our asymptotic notation.}. We will focus on the following problem.
	\begin{definition}\label{def:LTD}[\textbf{Triangle Membership}].
		In the triangle membership problem, each node needs to detect whether it is a part of a triangle.
	\end{definition}
	We remark that our second lower bound also applies to the \emph{triangle detection} problem, where it is sufficient that \emph{some} node learns that there is a triangle, without being able to tell who participates.
	This is also the guarantee given by the sublinear-round algorithm of Izumi and Le Gall~\cite{IzumiG17}.
	
\section{A Bandwidth Lower Bound for \texorpdfstring{$1$}{1}-round Triangle Membership}\label{sec:bandwidth}

In this section we show the following theorem.
\begin{theorem}~\label{thm: Bandwidth}
	The triangle membership problem cannot be solved by a single-round algorithm in the \congest{B} model unless $B \geq \log((\frac{n-2}{2(\Delta-1)})^{\Delta-1})$.
\end{theorem}
Observe that this implies an $\Omega(\Delta\log n)$ lower bound on $B$ for $\Delta=O(n^{1-\epsilon})$, and an $\Omega(n)$ lower bound on $B$ for $\Delta=\frac{n-2}{4}+1$.
This lower bound is significantly easier to obtain than the one in Section~\ref{sec:OneBit}.
This has the advantage of demonstrating the technique in a simpler context before delving into the proof of Theorem~\ref{thm: main - triangles membership}.

The main line of proof is to show that if the size of the messages is less than $\log{(\frac{n-2}{2(\Delta-1)})^{\Delta-1}}$, then there are three nodes $u,w,v\in V$, such that if the only neighbors of $u$ are $v$ and $w$, then $u$ receives the same messages from $v$ and $w$ during the single communication round regardless of whether $v$ and $w$ are connected (see Figure~\ref{fig:indistinguishability}). This is the standard notion of indistinguishability, given here under bandwidth restrictions for the first time.

\begin{figure}[ht]
	\centering
	\includegraphics[width=\textwidth, trim={0 5.5cm 0 6cm}, clip]{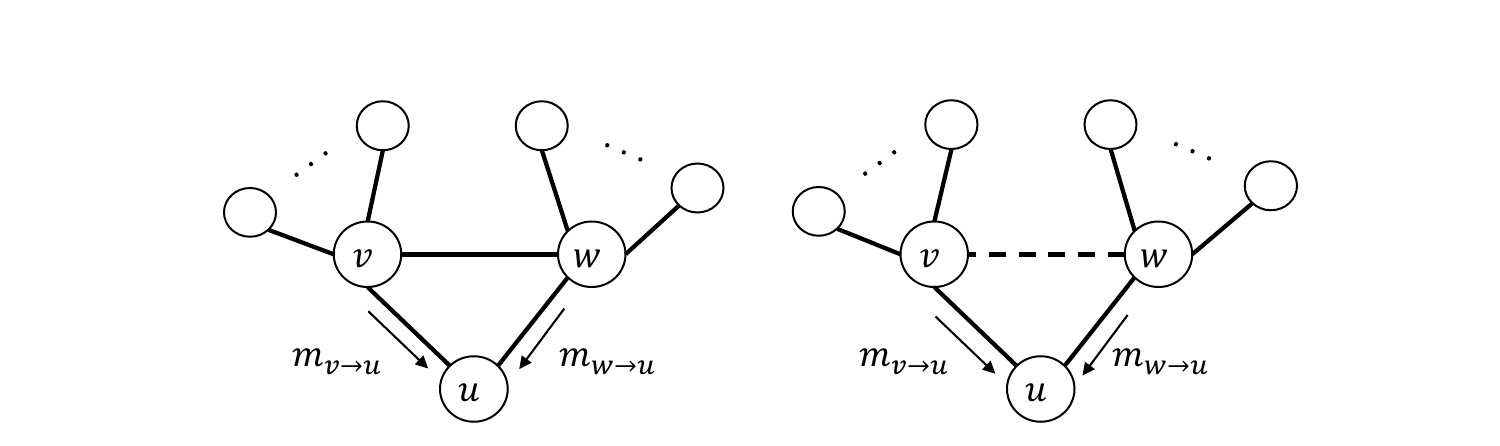}
	\caption{The node $u$ receives the same messages regardless of whether $v$ and $w$ are connected.}
	\label{fig:indistinguishability}
\end{figure}

For simplicity, we assume that each node has exactly $\Delta$ neighbors, except some special node $u$, which is fixed in the rest of this section, which has only two neighbors. Denote by $N(v)$ the set of neighbors of a node $v$, and let $m_{v\rightarrow u}(S)$ denote the message sent from $v$ to $u$ during the single round, given that $N(v) = S \cup \{u\}$.

The following two notions of fooling sets of nodes and fooling nodes are what allows us to capture indistinguishability in this setting. These are the specific shapes that our notion of \emph{fooling views} takes for obtaining this result.

\begin{definition}\label{def:foolingSetsNodes}
	Let $v\in V\setminus\{u\}$. A set of nodes $S$ is called a \emph{$(v,u)$-fooling set} if there is another set of nodes $S'\neq S$ such that $m_{v\rightarrow u}(S) = m_{v\rightarrow u}(S')$.
	
	A node $w\in V\setminus \{v,u\}$ is called a \emph{$(v,u)$-fooling node} if there are two sets of nodes $S\neq S'$, such that $w\in S$ and $w\notin S'$, and  $m_{v\rightarrow u}(S) = m_{v\rightarrow u}(S')$.
\end{definition}

We denote by $F_{nodes}(v,u)$ the set of $(v,u)$-fooling nodes. 

Our first step towards proving Theorem~\ref{thm: Bandwidth} is to show that for each $v\in V\setminus \{u\}$ there are many $(v,u)$-fooling nodes.

\begin{lemma}[\textbf{Many fooling nodes}]\label{ob: fooling sets}
	If $B < \log((\frac{n-2}{2(\Delta-1)})^{\Delta-1})$, then for each $v\in V\setminus \{u\}$ it holds that $|F_{nodes}(v,u)|\geq \frac{n-2}{2}$.
\end{lemma}

\begin{proof}
	Assume towards a contradiction that $|F_{nodes}(v,u)|<\frac{n-2}{2}$. Thus, there are $k'\geq\frac{n-2}{2}$ non $(v,u)$-fooling nodes, denote this set of non $(v,u)$-fooling nodes by $NF_{nodes}(v,u)$, and denote the family of sets of size $\Delta-1$ over nodes in $NF_{nodes}(v,u)$ by $NF_{sets}(v,u)$. It holds that
	$$|NF_{sets}(v,u)|={|NF_{nodes}(v,u)|\choose \Delta-1}\geq {\frac{n-2}{2}\choose \Delta-1}\geq\left(\frac{n-2}{2(\Delta-1)}\right)^{\Delta-1}$$ 
	
	Observe that $v$ must send a unique message to $u$ on each of these sets, since otherwise, there are two sets $S_1\neq S_2\in NF_{sets}(v,u)$ such that $m_{v\rightarrow u}(S_1)=m_{v\rightarrow u}(S_2)$, and therefore, by Definition~\ref{def:foolingSetsNodes}, at least one node in $NF_{nodes}(v,u)$ is $(v,u)$-fooling. This implies that $B\geq \log((\frac{n-2}{2(\Delta-1)})^{\Delta-1})$, a contradiction. 
\end{proof}

The strength of having many $(v,u)$-fooling nodes for every node $v$ is that it implies that there is a node $w^*$ that is a $(v,u)$-fooling node for many nodes $v$.

\begin{lemma}\label{cor:SuperFND}
	If $B < \log((\frac{n-2}{2(\Delta-1)})^{\Delta-1})$, then there is a node $w^*\in V$ and a set of nodes $P_{w^*}$ of size $\frac{n-2}{2}$, such that for each $v\in P_{w^*}$, it holds that $w^*\in F_{nodes}(v,u)$.
\end{lemma}

\begin{proof}
	For each $w,v\in V$, let $X_{w,v}$ be an indicator defined as follows:
	\begin{align*}X_{w,v} = \begin{cases}
			1 & \text{if } w\in F_{nodes}(v,u)\\
			0 & otherwise
		\end{cases}
	\end{align*}
	
	Since
	\begin{align*}
		\sum_{w\in V\setminus \{u\}} \sum_{v\in V\setminus \{u\}} X_{w,v} =
		\sum_{v\in V\setminus \{u\}} \sum_{w\in V\setminus \{u\}} X_{w,v},			
	\end{align*}
	
	Lemma~\ref{ob: fooling sets} gives that:
	\begin{align*}
		\sum_{w\in V\setminus \{u\}} \sum_{v\in V\setminus \{u\}} X_{w,v} \geq
		(n-1)\cdot\frac{n-2}{2}
	\end{align*}
	Therefore, there must be a node $w^*$ and a set of nodes $P_{w^*}$ of size $\frac{n-2}{2}$, such that for each $v\in P_{w^*}$, it holds that $X_{w,v}=1$. \footnote{In fact, one can show that there are many such nodes as $w^*$, but we do not use this in this section.}
\end{proof}

Lemma~\ref{cor:SuperFND} is what allows us to prove that $u$ cannot solve triangle membership, as follows.

\begin{proofof}{Theorem~\ref{thm: Bandwidth}}
	Let $w^*$ be the node provided by Lemma~\ref{cor:SuperFND}. It holds that $\binom{|P_{w^*}|}{\Delta-1}\geq (\frac{n-2}{2(\Delta-1)})^{\Delta-1}$. If $w^*$ sends to $u$ less than $\log{(\frac{n-2}{2(\Delta-1)})^{\Delta-1}}$ bits, then there are two sets of nodes $\{v_1,...,v_{\Delta-1}\}\neq \{v'_1,...,v'_{\Delta-1}\}\subseteq P_{w^*}$, such that $m_{{w^*}\rightarrow u}(\{v_1,...,v_{\Delta-1}\})=m_{{w^*}\rightarrow u}(\{v'_1,...,v'_{\Delta-1}\})$. Let $v^*$ be a node such that $v^*\in \{v_1,...,v_{\Delta-1}\}$ and $v^*\notin \{v'_1,...,v'_{\Delta-1}\}$. Since $v^*\in P_{w^*}$, by Definition~\ref{def:foolingSetsNodes}, it holds that there are two sets of nodes $S\neq S'$, such that $w^*\in S$ and $w^*\notin S'$, and $m_{v^*\rightarrow u}(S)=m_{v^*\rightarrow u}(S')$.
	
	To summarize, if the size of the messages is less than $\log{(\frac{n-2}{2(\Delta-1)})^{\Delta-1}}$, then there are two nodes $w,v\in V\setminus \{u\}$, such that if $u$ is connected to both $v$ and $w$, then $u$ receives the same messages from $v$ and $w$ during the single communication round regardless of whether $v$ and $w$ are connected (see Figure~\ref{fig:indistinguishability}).
\end{proofof}

\section{A Round Lower Bound for Triangle Detection with \texorpdfstring{$B=1$}{B = 1}}
\label{sec:OneBit}

Our main goal in this section is to prove the following theorem.

\begin{theorem}\label{thm: main - triangles membership}
	Any algorithm for triangle membership in the \congest{1} model requires $\Omega(\log^*n)$ rounds, even for graphs of maximum degree 2.
\end{theorem}

The hard instances for this setting are as simple as tripartite graphs with degree 2, as follows.
Let $\mathcal{G}=\{(V=A_1\dot\cup A_2\dot\cup A_3,E)\mid E\subseteq (A_1\times A_2)\cup (A_1\times A_3)\cup (A_2\times A_3)\}$ be the family of tripartite graphs such that for each $G\in \mathcal{G}$, it holds that $|A_1|=|A_2|=|A_3|=n/3$, and for every $i \in \{1,2,3\}$, each node in $A_i$ has exactly one neighbor in each of $A_{j}$ and $A_{k}$, where $j\neq k\neq i \in \{1,2,3\}$.

In a nutshell, we show that there must be 6 nodes that cannot detect whether they constitute a $6$-cycle or two triangles. Remarkably, this simple intuitive task requires us to develop novel machinery and make use of known results in extremal graph theory. Our approach has the compelling feature that -- up to one round -- in this setting triangle detection is equivalent to triangle membership, as a node detecting a triangle can infer that it is part of a triangle itself and just needs to inform its neighbors.
Thus, we obtain the same lower bound for the, in general, possibly easier problem of triangle detection. Moreover, the lower bound depends only on the set of \emph{possible} IDs, not the actual number of nodes in the graph.

\begin{theorem}\label{thm:namespace}
	Any algorithm for triangle detection in the \congest{1} model requires $\Omega(\log^*N)$ rounds, even for graphs of maximum degree 2 with only 6 nodes, given a namespace of size $N$.
\end{theorem}

The main challenge in proving a lower bound for more than one round is that for $t > 1$ the communication between the nodes in round $t$ does not depend only on their IDs and their sets of neighbors (as is the case for the first round), but rather it depends also on their \emph{views} after $t-1$ rounds, i.e., on the set of messages that the nodes receive during each of the first $t-1$ rounds. The view of a node after $t-1$ rounds may depend on all the nodes in its $(t-1)$-hop neighborhood.

Nevertheless, since the family of graphs $\mathcal{G}$ is defined such that, for each $G=(A_1\dot\cup A_2\dot\cup A_3,E)\in \mathcal{G}$, it holds that each node in each part of $G$ has exactly one neighbor in each of the other two parts, any cycle in any $G\in \mathcal{G}$ must be a connected component. Therefore, if there is a triple $(u,w,x)\in A_1\times A_2\times A_3$ that is connected by a triangle $(u,w,x)$ in $G$, then the communication between the nodes $u,w$ and $x$ during any round in $G$ depends only on the nodes $u,w$ and $x$. Similarly, if there is a 6-cycle $(u_1,w_1,x_1,u_2,w_2,x_2)$ in $G$, then the communication between the nodes $u_1,w_1,x_1,u_2,w_2,x_2$ during any round in $G$ depends only on the nodes $u_1,w_1,x_1,u_2,w_2,x_2$.

Our main line of proof is to show that for any algorithm for triangle membership with $B=1$, there exist $u_1,w_1,x_1,u_2,w_2,x_2\in A_1\dot\cup A_2\dot\cup A_3$ such that $u_1$ receives the same messages during the first $\Omega(\log^*n)$ rounds in two different scenarios, the first of which is a scenario in which $u_1$ participates in a triangle $(u_1,x_1,w_1)$, and the second is a scenario in which $u_1$ participates in a 6-cycle $(u_1,x_1,w_2,u_2,x_2,w_1)$. See Figure~\ref{fig:3x6} for an illustration.

Given three nodes $u,w,x\in A_1\dot\cup A_2\dot\cup A_3$, we denote by $m^t_{w\rightarrow u}((u,w,x))$ the message sent from $w$ to $u$ during round $t$, given that $u,w$ and $x$ are connected by a triangle. Similarly, given six nodes $u_1,w_1,x_1,u_2,w_2,x_2\in A_1\dot\cup A_2\dot\cup A_3$, we denote by $m^t_{w_1\rightarrow u_1}((u_1,w_1,x_2,u_2,w_2,x_1))$ the message sent from $w_1$ to $u_1$ during round $t$, given that the nodes $u_1,w_1,x_2,u_2,w_2,x_1$ are connected by a 6-cycle $(u_1,x_1,w_2,u_2,x_2,w_1)$.

The structure of the proof for the lower bound on the number of rounds is as follows. In Section~\ref{sec:fooling sets of triangles} we present the notion of \emph{fooling sets of triangles} which shares a similar spirit to Definition~\ref{def:foolingSetsNodes} and is the basis of our fooling views technique in this section. Next, in section~\ref{sec: triangles 6-cycles} we present the connection to 6-cycles and deduce our main result.

\begin{figure}
	\centering
	\includegraphics[width=\textwidth, trim={0 7cm 2cm 9cm}, clip]{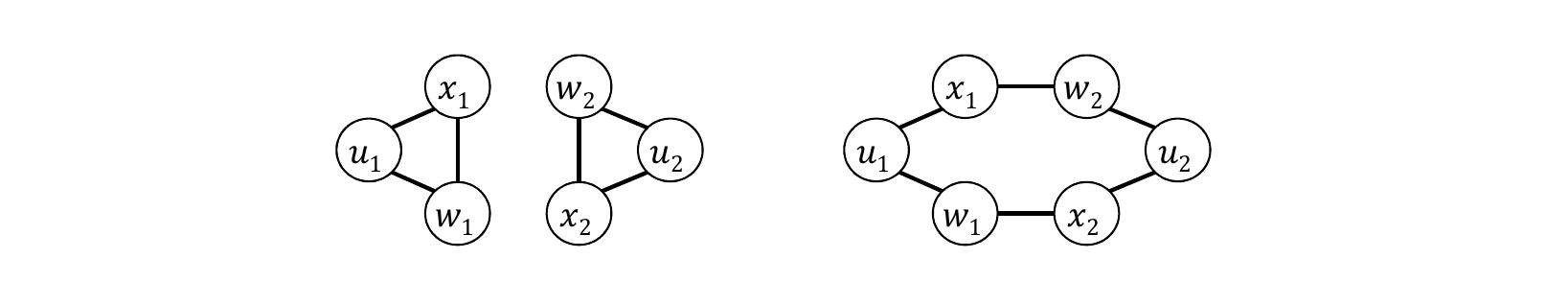}
	\caption{The structures we use are the triangle $(u_1,x_1,w_1)$ and the 6-cycle $(u_1,x_1,w_2,u_2,x_2,w_1)$.}
	\label{fig:3x6}
\end{figure}

\subsection{Fooling sets of triangles}\label{sec:fooling sets of triangles}

\begin{definition}\label{def: t-rounds fooling set}
	Fix a pair of nodes $(u,w)\in (A_i \times A_j) \cap E$, for $i \neq j$, and let $X=\{x_1,...,x_{|X|}\}\subseteq A_k$ for $k \neq i,j$. A set of triangles $\{(u,w,x)\mid x\in X\}$ is called a \emph{$(w,u)^t$-fooling set of triangles} if $w$ sends the same message to $u$ during each of the first $t$ rounds, in each of the triangles in $\{(u,w,x)\mid x\in X\}$. Formally, for each $1\leq i\leq t$,
	\center{
	$m^i_{w\rightarrow u}((u,w,x_1))=m^i_{w\rightarrow u}((u,w,x_2))=...=m^i_{w\rightarrow u}((u,w,x_{|X|})).$
	}
\end{definition}

We extend the notion of a fooling set of triangles to the notion of a \emph{fooling rectangle of triangles}:

\begin{definition}\label{def: t-rounds fooling rectangle}
	Fix a node $u\in A_i$. A rectangle $W\times X\subseteq A_j\times A_k$, for $j \neq k \neq i$, is called a \emph{$u^t$-fooling rectangle of triangles} if it satisfies the following properties.
	\begin{enumerate}
		\item For each $w\in W$, it holds that $\{(u,w,x)\mid x\in X\}$ is a $(w,u)^t$-fooling set of triangles.
		\item For each $x\in X$, it holds that $\{(u,w,x)\mid w\in W\}$ is an $(x,u)^t$-fooling set of triangles.
	\end{enumerate}
\end{definition}

Finally, we extend the notion of a fooling rectangle of triangles to a \emph{fooling cube of triangles}:

\begin{definition}\label{def: t-round fooling tripartite}
	A cube $U\times W\times X\subseteq A_1\times A_2\times A_3$ is called a \emph{$t$-fooling cube of triangles} if:
	\begin{enumerate}
		\item For each $u\in U$, it holds that $W\times X$ is a $u^t$-fooling rectangle of triangles.
		\item For each $w\in W$, it holds that $U\times X$ is a $w^t$-fooling rectangle of triangles.
		\item For each $x\in X$, it holds that $U\times W$ is a $x^t$-fooling rectangle of triangles.
	\end{enumerate}
\end{definition}
The following observation follows immediately from Definition~\ref{def: t-round fooling tripartite}.

\begin{observation}\label{ob: zero cube}
	$A_1\times A_2\times A_3$ is a $0$-fooling cube of triangles.
\end{observation}

Our next step towards proving Theorem~\ref{thm: main - triangles membership} is to show that any $t$-fooling cube of triangles of size $s$ contains a $(t+1)$-fooling cube of triangles of size $\Omega(\sqrt{\log\log\log s})$. 
\begin{lemma}\label{lemma: multy-round triangles}
	If there is a $t$-fooling cube of triangles $U\times W\times X\subseteq A_1\times A_2\times A_3$ of size $|U|=|W|=|X|=s$, then there is a constant $\beta$ and a $(t+1)$-fooling cube of triangles $U'\times W'\times X'\subseteq U\times W\times X$ of size $|U'|=|W'|=|X'|=\beta\sqrt{\log\log\log s}$.
\end{lemma}

The key combinatorial ingredient for proving Lemma~\ref{lemma: multy-round triangles} is the following corollary of Theorem~\ref{thm: large hypergraphs1}.
\begin{corollary}[of Theorem~\ref{thm: large hypergraphs1}]
	\label{cor: large hypergraphs}
	For every constant $\alpha$, for sufficiently large $s$, any $s\times s\times s$ boolean cube that contains at least $e^{-1/\alpha^2}s^3$ entries that are $1$, contains a $1$-monochromatic subcube, such that each side of the subcube is of size $\alpha\sqrt{\log s}$.
\end{corollary}

To prove Lemma~\ref{lemma: multy-round triangles}, we first prove a weaker claim, which guarantees the existence of a sufficiently large cube of triangles satisfying only the first property in Definition~\ref{def: t-round fooling tripartite}. That is, we prove the existence of a cube $U'\times W'\times X'\subseteq U\times W\times X$ such that for each $u\in U'$, it holds that $W'\times X'$ is a $u^{t+1}$-fooling rectangle of triangles and $|U'|=|W'|=|X'|=\Theta(\sqrt{\log n} )$. Then, applying this claim to the 3 different sides of the cube $U\times W\times X$ gives Lemma~\ref{lemma: multy-round triangles}.
	
	\begin{claim}\label{claim: helper for main lemma multy-round triangles}
		If there is a $t$-fooling cube of triangles $U\times W\times X\subseteq A_1\times A_2\times A_3$ of size $|U|=|W|=|X|=s$, then there is a constant $\gamma$ and a cube $U'\times W'\times X'\subseteq U\times W\times X$ of size $|U'|=|W'|=|X'|=\gamma\sqrt{\log s}$, such that for each $u\in U'$, it holds that $W'\times X'$ is a $u^{t+1}$ fooling rectangle of triangles.
		\begin{proof}
			Fix $u\in U$, and consider the messages that $u$ receive from its neighbors in $W\cup X$ during round $t+1$. Since $B=1$, for each $w\in W$, it holds that there exist $s/2$ nodes $x_1,...,x_{s/2} \in X$ such that $m^{t+1}_{w\rightarrow u}((u,w,x_1))=m^{t+1}_{w\rightarrow u}((u,w,x_2))=...=m^{t+1}_{w\rightarrow u}((u,w,x_{s/2}))$. Denote this specific message by $\tilde{m}^{t+1}_{w\rightarrow u}$. Therefore, defining an indicator variable $Y^{t+1}_{u,w,x}$ for each $(w,x)\in W\times X$ as follows
			\begin{align*}Y^{t+1}_{u,w,x} = \begin{cases}
					1 & \text{if } m^{t+1}_{w\rightarrow u}((u,w,x))=\tilde{m}^{t+1}_{w\rightarrow u}\\
					0 & otherwise
				\end{cases}
			\end{align*}
			gives that
			\begin{align*}
				\sum_{w\in W} \sum_{x\in X} Y^{t+1}_{u,w,x}=s^2/2,
			\end{align*}
			which implies
			\begin{align*}
				\sum_{x\in X} \sum_{w\in W} Y^{t+1}_{u,w,x}=\sum_{w\in W} \sum_{x\in X} Y^{t+1}_{u,w,x}=s^2/2.
			\end{align*}
			
			It follows that there are at least $s/4$ nodes $x_1,...,x_{s/4}$ in $X$, such that each $x_i\in \{x_1,...,x_{s/4}\}$ has a set of nodes $P(x_i)\subseteq W$ of size $s/4$, such that for each $w\in P(x_i)$, it holds that $Y^{t+1}_{u,w,x_i}=1$. Furthermore, for each $x_i$, it holds that there are $|P(x_i)|/2$ nodes $w^{x_i}_1,...,w^{x_i}_{|P(x_i)|/2} \in P(x_i)$ such that $m^{t+1}_{x_i\rightarrow u}((u,w^{x_i}_1,x_i))=...=m^{t+1}_{x_i\rightarrow u}((u,w^{x_i}_{|P(x_i)|/2},x_i))$. Denote this specific message by $\tilde{m}^{1}_{x_i\rightarrow u}$.
			
			For each $(u,w,x)\in U\times W\times X$, let $Z^{t+1}_{u,w,x}$ be an indicator defined as follows:
			\begin{align*}Z^{t+1}_{u,w,x} = \begin{cases}
					1 & \text{if } (m^{t+1}_{w\rightarrow u}((u,w,x))=\tilde{m}^{t+1}_{w\rightarrow u})\wedge (m^{t+1}_{x\rightarrow u}((u,w,x))=\tilde{m}^{t+1}_{x\rightarrow u})\\
					0 & otherwise
				\end{cases}
			\end{align*}
			
			Notice that $Z^{t+1}_{u,w,x}$ is a boolean cube of size $s\times s \times s$, and that the above argument implies that it contains at least $s^3/32$ entries that are $1$.
			
			Therefore, by Corollary~\ref{cor: large hypergraphs}, the cube $Z^{t+1}_{u,w,x}$ contains a $1$-monochromatic subcube such that each side of the subcube is of size $\sqrt{\log s/\log32}$. Denote this subcube by $U'\times W' \times X'$, and the claim follows.
		\end{proof}
	\end{claim}
	
\begin{proofof}{Lemma~\ref{lemma: multy-round triangles}}
	To finish the proof of Lemma~\ref{lemma: multy-round triangles}, we apply Claim~\ref{claim: helper for main lemma multy-round triangles} on each of the three sides of the $t$-fooling cube of triangles $U\times W\times X$, and we deduce that there is a constant $\beta$ and a $(t+1)$-fooling cube of triangles $U'\times W'\times X'\subseteq U\times W\times X$ of size $|U'|=|W'|=|X'|=\beta\sqrt{\log\log\log s}$.
\end{proofof}

What we have so far, by Observation~\ref{ob: zero cube} and Lemma~\ref{lemma: multy-round triangles}, is that there are sufficiently many triangles in which the nodes receive the same messages from their neighbors. What remains is to actually capture the two different scenarios of a node participating in a triangle or not, while having the node keep receiving the same messages in both. For this we next discuss 6-cycles.

\subsection{Triangles and 6-Cycles}\label{sec: triangles 6-cycles}

Having defined fooling structures of triangles in the previous section, our final step towards proving Theorem~\ref{thm: main - triangles membership} is presenting a connection between fooling structures of triangles and $6$-cycles. We start with the following definitions.

\begin{definition}\label{def: triangle 6-cycle}
	Let $(u_1,w_1,x_1),(u_2,w_2,x_2)\in A_i\times A_j\times A_k$ be two disjoint triples, for $i\neq j\neq k\in \{1,2,3\}$. The triangle $(u_1,w_1,x_1)$ and the 6-cycle $(u_1,w_1,x_2,u_2,w_2,x_1)$ are called \emph{$(3\leftrightarrow 6,u_1)^t$-fooling}, if $u_1$ receives the same messages from $w_1$ and $x_1$ during each of the first $t$ rounds, in the triangle $(u_1,w_1,x_1)$ and in the 6-cycle $(u_1,w_1,x_2,u_2,w_2,x_1)$. Formally, for each $1\leq i\leq t$,
		\center
		{
		$m^i_{w_1\rightarrow u_1}((u_1,w_1,x_1))=m^i_{w_1\rightarrow u_1}((u_1,w_1,x_2,u_2,w_2,x_1))$\\
		$\wedge m^i_{x_1\rightarrow u_1}((u_1,w_1,x_1))=m^i_{x_1\rightarrow u_1}((u_1,w_1,x_2,u_2,w_2,x_1)).$
		}
\end{definition}

To capture the connection between triangles and 6-cycles, we extend the notion of a $t$-fooling cube of triangles into the notion of a $(3\leftrightarrow 6)^t$-fooling cube.
\begin{definition}\label{def: main cube triangle 6-cycle}
	A cube $U\times W\times X\subseteq A_1\times A_2\times A_3$ is called a \emph{$(3\leftrightarrow 6)^t$-fooling cube}, if:
	\begin{enumerate}
		\item The cube $U\times W\times X$ is a $t$-fooling cube of triangles (see Definition~\ref{def: t-round fooling tripartite}).
		\item\label{prop:3-6} For each pair of disjoint triples $(u_1,w_1,x_1),(u_2,w_2,x_2)\subseteq U\times W\times X$ it holds that:
		\begin{enumerate}
			\item\label{def: main property u} The triangle $(u_1,w_1,x_1)$ and the 6-cycle $(u_1,w_1,x_2,u_2,w_2,x_1)$ are $(3\leftrightarrow 6,u_1)^t$-fooling (see Definition~\ref{def: triangle 6-cycle}).
			\item\label{def: main property w} The triangle $(u_1,w_1,x_1)$ and the 6-cycle $(u_1,w_1,x_1,u_2,w_2,x_2)$ are $(3\leftrightarrow 6,w_1)^t$ fooling.
			\item\label{def: main property x} The triangle $(u_1,w_1,x_1)$ and the 6-cycle $(u_1,w_2,x_2,u_2,w_1,x_1)$ are $(3\leftrightarrow 6,x_1)^t$-fooling.
		\end{enumerate}
	\end{enumerate}
\end{definition}

Our goal next is to show that there is a $(3\leftrightarrow 6)^{\Omega(\log^*n)}$-fooling cube $U\times W\times X$ of size $|U|=|W|=|X|\geq 2$. The following observation follows immediately from definition~\ref{def: main cube triangle 6-cycle}.

\begin{observation}\label{ob: first round fooling triangle 6-cycle}
	$A_1\times A_2\times A_3$ is a $(3\leftrightarrow 6)^0$-fooling cube of size $n/3$.
\end{observation}

In general, in life, it is good to know that there is a $(3\leftrightarrow 6)^0$-fooling cube somewhere in the wild. However, as we are interested in a lower bound on the number of rounds, we need to show that the amazing $(3\leftrightarrow 6)^0$-fooling cube contains \emph{sufficiently large} sets of fooling cubes \emph{during many rounds}. For this, we need to prove the following lemma.

\begin{lemma}\label{lemma: main multi-round}
	If there is a $(3\leftrightarrow 6)^t$-fooling cube $U\times W\times X\subseteq A_1\times A_2\times A_3$ of size $|U|=|W|=|X|=s$, then there is a $(3\leftrightarrow 6)^{t+1}$-fooling cube $U'\times W'\times X'\subseteq U\times W\times X$ of size $|U'|=|W'|=|X'|=\Omega(\sqrt{\log\log\log s})$.
\end{lemma}

\begin{proof}
	By the first property of Definition~\ref{def: main cube triangle 6-cycle}, the cube $U\times W\times X$ is a $t$-fooling cube of triangles. Therefore, by
	Lemma~\ref{lemma: multy-round triangles},
	there is a $(t+1)$-fooling cube of triangles $U'\times W'\times X'\subseteq U\times W\times X$ of size $|U'|=|W'|=|X'|=\Omega(\sqrt{\log\log\log s})$. We show that $U'\times W'\times X'$ is also a $(3\leftrightarrow 6)^{t+1}$-fooling cube. That is, we show that each pair of disjoint triples $(u_1,w_1,x_1),(u_2,w_2,x_2)\in U'\times W'\times X'$ satisfies Properties~\ref{def: main property u},~\ref{def: main property w} and ~\ref{def: main property x} of Definition~\ref{def: main cube triangle 6-cycle}.
	
	We start with Property~\ref{def: main property u}. That is, we first prove that for each pair of disjoint triples $(u_1,w_1,x_1),(u_2,w_2,x_2)\in U'\times W'\times X'$ it holds that the triangle $(u_1,w_1,x_1)$ and the 6-cycle $(u_1,w_1,x_2,u_2,w_2,x_1)$ are $(3\leftrightarrow 6,u_1)^{t+1}$-fooling. Observe that since $U'\times W'\times X'$ is a $(3\leftrightarrow 6)^t$-fooling cube, by Property~\ref{def: main property u} of Definition~\ref{def: main cube triangle 6-cycle}, it holds that for each $1\leq i\leq t$,
	\begin{align*}
	&m^{i}_{w_1\rightarrow u_1}((u_1,w_1,x_1))=m^{i}_{w_1\rightarrow u_1}((u_1,w_1,x_2,u_2,w_2,x_1))\\
	\wedge&m^{i}_{x_1\rightarrow u_1}((u_1,w_1,x_1))=m^{i}_{x_1\rightarrow u_1}((u_1,w_1,x_2,u_2,w_2,x_1)).
	\end{align*}
	Therefore, in order to show that the triangle $(u_1,w_1,x_1)$ and the 6-cycle $(u_1,w_1,x_2,u_2,w_2,x_1)$ are $(3\leftrightarrow 6,u_1)^{t+1}$-fooling, it remains to show that the above holds also for $t+1$. That is, we need to show that
	\begin{align*}
	&m^{t+1}_{w_1\rightarrow u_1}((u_1,w_1,x_1))=m^{t+1}_{w_1\rightarrow u_1}((u_1,w_1,x_2,u_2,w_2,x_1))\\
	\wedge&m^{t+1}_{x_1\rightarrow u_1}((u_1,w_1,x_1))=m^{t+1}_{x_1\rightarrow u_1}((u_1,w_1,x_2,u_2,w_2,x_1)).
	\end{align*}
	Observe that since $U'\times W'\times X'$  is a $(t+1)$-fooling cube of triangles, it holds that
	\begin{align}
	&\label{eq: w_1}m^{t+1}_{w_1\rightarrow u_1}((u_1,w_1,x_1))=m^{t+1}_{w_1\rightarrow u_1}((u_1,w_1,x_2))\\
	\wedge&\label{eq: x_1} m^{t+1}_{x_1\rightarrow u_1}((u_1,w_1,x_1))=m^{t+1}_{x_1\rightarrow u_1}((u_1,w_2,x_1)).
	\end{align}
	Furthermore, Since $U'\times W'\times X'$ is also a $(3\leftrightarrow 6)^t$-fooling cube, by Property~\ref{def: main property w} of Definition~\ref{def: main cube triangle 6-cycle}, for the two triples $(u_1,w_1,x_2),(u_2,w_2,x_1)$, it holds that for each $1\leq i\leq t$,
	\begin{align*}
	&m^{i}_{u_1\rightarrow w_1}((u_1,w_1,x_2))=m^{i}_{u_1\rightarrow w_1}((u_1,w_1,x_2,u_2,w_2,x_1))\\
	\wedge&m^{i}_{x_2\rightarrow w_1}((u_1,w_1,x_2))=m^{i}_{x_2\rightarrow w_1}((u_1,w_1,x_2,u_2,w_2,x_1)),
	\end{align*}
	which means that $w_1$ has the same view after $t$ rounds both in the case that it participates in a triangle $((u_1,w_1,x_2))$, and in the case that it participates in a 6-cycle $((u_1,w_1,x_2,u_2,w_2,x_1))$. This, in turn, implies that it sends the same message to $u_1$ during round $t+1$, in these two scenarios, that is:
	\begin{align*}
	&m^{t+1}_{w_1\rightarrow u_1}((u_1,w_1,x_2))=m^{t+1}_{w_1\rightarrow u_1}((u_1,w_1,x_2,u_2,w_2,x_1)).
	\end{align*}
	Combining this with Equation (\ref{eq: w_1}), gives that
	\begin{align*}
	&m^{t+1}_{w_1\rightarrow u_1}((u_1,w_1,x_1))=m^{t+1}_{w_1\rightarrow u_1}((u_1,w_1,x_2,u_2,w_2,x_1)).
	\end{align*}
	
	Similarly, by Property~\ref{def: main property x} of Definition~\ref{def: main cube triangle 6-cycle}, for the two triples $(u_1,w_2,x_1),(u_2,w_1,x_2)$, it holds that for each $1\leq i\leq t$,
	\begin{align*}
	&m^{i}_{u_1\rightarrow x_1}((u_1,w_2,x_1))=m^{i}_{u_1\rightarrow x_1}((u_1,w_1,x_2,u_2,w_2,x_1))\\
	\wedge& m^{i}_{w_2\rightarrow x_1}((u_1,w_2,x_1))=m^{i}_{w_2\rightarrow x_1}((u_1,w_1,x_2,u_2,w_2,x_1)),
	\end{align*}
	which implies that
	\begin{align*}
	&m^{t+1}_{x_1\rightarrow u_1}((u_1,w_2,x_1))=m^{t+1}_{x_1\rightarrow u_1}((u_1,w_1,x_2,u_2,w_2,x_1)).
	\end{align*}
	Combining this with Equation (\ref{eq: x_1}), gives that
	\begin{align*}
	&m^{t+1}_{x_1\rightarrow u_1}((u_1,w_1,x_1))=m^{t+1}_{x_1\rightarrow u_1}((u_1,w_1,x_2,u_2,w_2,x_1)),
	\end{align*}
	which completes the proof that for each pair of disjoint triples $(u_1,w_1,x_1),(u_2,w_2,x_2)\in U'\times W'\times X'$ it holds that the triangle $(u_1,w_1,x_1)$ and the 6-cycle $(u_1,w_1,x_2,u_2,w_2,x_1)$ are $(3\leftrightarrow 6,u_1)^{t+1}$-fooling, i.e., the cube $U'\times W'\times X'$ satisfies Property~\ref{def: main property u} of Definition~\ref{def: main cube triangle 6-cycle}. By symmetric arguments, $U'\times W'\times X'$ also satisfies Properties~\ref{def: main property w} and~\ref{def: main property x} of Definition~\ref{def: main cube triangle 6-cycle}.
\end{proof}

\begin{proofof}{Theorem~\ref{thm: main - triangles membership}}
	Observe that it is sufficient to prove that there is a $(3\leftrightarrow 6)^{\Omega(\log^*n)}$-fooling cube $U\times W\times X\subseteq A_1\times A_2\times A_3$ of size $|U|=|W|=|X|\geq 2$. By Observation~\ref{ob: first round fooling triangle 6-cycle}, $A_1\times A_2\times A_3$ is a $(3\leftrightarrow 6)^0$-fooling cube  of size $n/3$. Furthermore, by Lemma~\ref{lemma: main multi-round}, for any $t\geq 0$, if there is a $(3\leftrightarrow 6)^t$-fooling cube $U^t\times W^t\times X^t$ of size $|U^t|=|W^t|=|X^t|=s$, then there is a  $(3\leftrightarrow 6)^{t+1}$-fooling cube $U^{t+1}\times W^{t+1}\times X^{t+1}\subseteq U^t\times W^t\times X^t$ of size $|U^{t+1}|=|W^{t+1}|=|X^{t+1}|=\Omega(\sqrt{\log\log\log s})$. Therefore, applying Lemma~\ref{lemma: main multi-round} repeatedly $\Omega(\log^*n)$ times implies that there is a $(3\leftrightarrow 6)^{\Omega(\log^*n)}$-fooling cube $U\times W\times X\subseteq A_1\times A_2\times A_3$ of size $|U|=|W|=|X|\geq 2$.
\end{proofof}

\section{Discussion and Open Questions}
\label{sec:discussion}
Being a first-step type of contribution, this work would not be complete without pointing out many additional open questions for which our fooling views technique is a possible candidate as the road for making progress.

First, we raise the question of whether our specification of the triangle membership problem has an inherently different complexity than that of triangle detection. The latter is the standard way of phrasing decision problems in distributed computing (\emph{everyone} outputs \textsc{NO} for a no instance, \emph{someone} outputs \textsc{YES} for a yes instance). We believe that the bound we give for single-round algorithms in Section~\ref{sec:bandwidth} should hold also for triangle detection, but this seems to require a deeper technical analysis. Our lower bound for the number of rounds in \congest{1} does hold also for triangle detection as explained in Section~\ref{sec:OneBit}. Note that the sublinear algorithm of~\cite{IzumiG17} solves triangle detection but it is not clear how to make it solve triangle membership. A closely-related problem is that of triangle listing, where all triangles need to be output. The work of~\cite{IzumiG17} also gives the first sublinear algorithm for triangle listing, completing in $O(n^{3/4}\log{n})$ rounds in the \congestSTD model, as well as an $\Omega(n^{1/3}/\log{n})$ lower bound (see also\cite{Pandurangan16}).
\begin{oq}
\label{oq:detection-membership}
Do the triangle membership, triangle detection, and triangle listing problems have different complexities in the \congestSTD model?
\end{oq}

Our results in this paper are for deterministic algorithms, but we do not see a technical obstacle in making them work for randomized algorithms as well.
\begin{oq}
\label{oq:deterministic-randomized}
Is the deterministic complexity of triangle membership/detection strictly larger than that of its randomized complexity in the \congestSTD model?
\end{oq}

Section~\ref{sec:bandwidth} gives a tight bound on the bandwidth required for an optimal-round algorithm. In Appendix~\ref{sec:kcycle-reduction} we address the bandwidth of optimal-round algorithms for $k$-cycles. This question can be asked about further problems, even those for which we do not know yet what the exact round complexity is, such as various symmetry breaking problems.
\begin{oq}
\label{oq:optimal-round}
What is the \emph{bandwidth complexity} of optimal-round distributed algorithms for various problems?
\end{oq}

Whether gaps between the \local and \congestSTD models occur for symmetry breaking problems is a central open question, which we hope our technique can shed light upon.
Prime examples are MIS and $(\Delta+1)$-Coloring, but it is not known whether there is indeed a gap. The reason that reductions from 2-party communication problems are provably incapable of proving lower bounds for these problems is that any partial solution that is obtained by a greedy algorithm is extendable into a valid solution for the entire graph, which means that one player can solve the problem for its set of nodes and deliver only the state of nodes on the boundary to the other player, for completing the task. Another way to see why arguing about communication only will not suffice here is to notice that simulating the sequential greedy algorithm requires in fact very little communication in total, despite taking many rounds.
\begin{oq}
\label{oq:symmetry-breaking}
For various symmetry breaking problems, is the complexity in the \congestSTD model strictly higher than its counterpart in the \local model?
\end{oq}

\paragraph{Acknowledgements:} We are grateful to Michal Dory, Eyal Kushilevitz, and Merav Parter for stimulating discussions. We are also grateful to Ivan Rapaport, Eric Remila, and Nicolas Schabanel, a point that was made by them helped in significantly simplifying Section~\ref{sec:bandwidth}.

\bibliographystyle{abbrv}
\bibliography{bib}

\appendix

\section{\texorpdfstring{$k$}{k}-Cycle Membership for \texorpdfstring{$k \geq 4$}{k > 3}}
\label{sec:kcycle-reduction}
An immediate question is whether we can get lower bounds on the bandwidth for additional round-optimal algorithms. We show here how to generalize our lower bound technique to apply for detecting membership in larger cycles. However, curiously, as we show in Section~\ref{subsec:kCycleCC} for the sake of comparison, for cycles that are larger than $3$, the standard approach of reductions from 2-party communication complexity problems allows for stronger lower bounds.

An optimal-round algorithm for solving $k$-cycle membership completes within exactly $\floor{(k-1)/2}$ rounds, by a simple (standard) indistinguishability argument (this is even with unlimited bandwidth).
This is because after $t$ rounds of communication a node may have information only about nodes at distance at most $t+1$ from it. Therefore, after $\floor{(k-1)/2}-1$ rounds of communication, a node cannot distinguish whether it participates in a $k$-cycle or not.

To see how we generalize our technique, observe first that the lower bound given in Theorem~\ref{thm: Bandwidth} holds even when a specific node $u$ is given as input to all the nodes, and only $u$ needs to solve the triangle membership problem. This is helpful in extending our lower bound to the case of $k$-cycles when $k \geq 4$, and hence we define this problem formally.

\begin{definition}\label{def:FLTD}(Fixed-Node Triangle Membership).
	In the Fixed-Node Triangle Membership problem, all nodes are given the identity of a specific node $u$, and node $u$ needs to detect whether it is a part of a triangle.
\end{definition}

The proof of Theorem~\ref{thm: Bandwidth} actually proves the following theorem.

\begin{theorem}\label{thm: FLTD}
	The fixed-node triangle membership problem cannot be solved by a single-round algorithm in the \congest{B} model unless $B \geq \log((\frac{n-2}{2(\Delta-1)})^{\Delta-1})$.
\end{theorem}

We formally extend the membership problem to larger cycles, as follows.
\begin{definition}\label{def: L4C}(k-Cycle Membership).
	In the $k$-Cycle Membership problem, each node needs to detect whether it is a part of a $k$-cycle.
\end{definition}

Now, Theorem~\ref{thm: FLTD} can be used to prove the following.

\begin{theorem}\label{thm:4-cycle}
	Let $k\in O(n^{1-\epsilon})$, for some constant $0<\epsilon<1$. The $k$-cycle membership problem cannot be solved by an optimal $\floor{(k-1)/2}$-round algorithm in the \congest{B} model unless $B \geq c\log((\frac{n-2}{2(\Delta-1)})^{\Delta-1})$, for some constant $c\geq 0$.
\end{theorem}

We prove Theorem~\ref{thm:4-cycle} by showing a reduction from the fixed-node triangle membership problem, given in Definition~\ref{def:FLTD}. That is, we show that an algorithm for solving the $k$-cycle membership problem can be used to solve the fixed-node triangle membership problem.

\begin{proofof}{Theorem~\ref{thm:4-cycle}}
	First, we show how to construct an appropriate instance for the $k$-cycle membership problem, given an instance for the fixed-node triangle membership problem. We start with describing the construction for an odd value of $k$, and then show how to tweak it to handle even  values of $k$ as well.
	
	Let $(G=(V,E),u)$ be an instance of the fixed-node triangle membership problem where $|V|=\widetilde{n}$, and let $E(u)$ be the set of edges incident to the node $u$. For an odd $k$, we define an instance of the $k$-cycle membership problem, $G'=(V',E')$, where $|V'|=n$, as follows. We replace each edge $(u,v)\in E(u)$ with a path $P_{uv}$ of length $(k-3)/2+1$, going through $(k-3)/2$ new nodes, denoted $uv_1,\dots, uv_{(k-3)/2}$, containing the edges $(u,uv_1), (uv_1,uv_2), \dots, (uv_{(k-3)/2-1},uv_{(k-3)/2}),(uv_{(k-3)/2},v)$.
	For example, for $k=7$, we replace each edge $(u,v)$ by a path of three edges $(u,uv_1),(uv_1,uv_2), (uv_2,v)$, going through two new intermediate nodes $uv_1$ and $uv_2$ (see Figure~\ref{fig: 4-cycle}).
	\begin{figure}
		\centering
		\includegraphics[scale=0.6, trim={0 6.5cm 9cm 7cm}, clip]{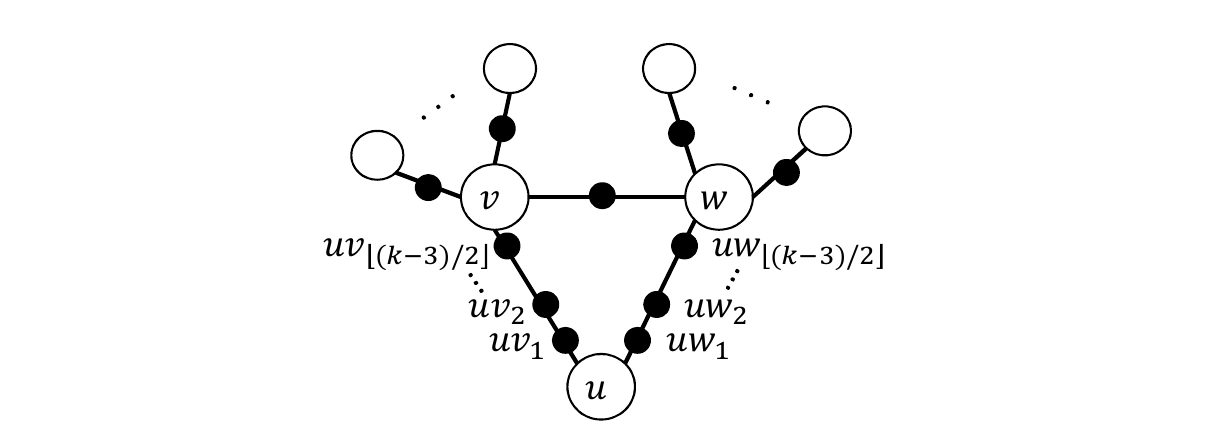}
		\caption{Constructing $G'$ out of $G$.}
		\label{fig: 4-cycle}
	\end{figure}
	Formally, the new graph $G'=(V',E')$ is defined as:
	\begin{align*}
		&V'=V\cup\{uv_i \mid (u,v)\in E(u), 1 \leq i \leq (k-3)/2\}\\
		&E'=(E\setminus E(u))\cup\{(u,uv_1)),(uv_1,uv_2), \dots, (uv_{(k-3)/2-1},uv_{(k-3)/2}), (uv_{(k-3)/2},v)\mid (u,v)\in E(u)\}\\
	\end{align*}	
	The following observation follows directly from the construction above.
	\begin{observation}\label{ob: triangle-4cycle reduction}
		For an odd $k$, the node $u$ participates in a $k$-cycle in $G'$ if and only if it participates in a triangle in $G$.
	\end{observation}
	
	Since an algorithm may use the IDs of the nodes, we need to also assign unique IDs to nodes in $G'$. We can do this in any consistent arbitrary manner, say, by assigning $ID_{G'}(x)=ID_{G}(x)$ if $x \in V$, and $ID(uv_{i}) = ID(v)\circ bin(i)$ for nodes in $V'\setminus V$, where $bin(i)$ is the binary representation of $i$. Notice that $\Delta(G')=\Delta(G)$.

	Next, assume towards a contradiction, that there is an algorithm $A'$ that solves the $k$-cycle membership problem in an optimal number of $(k-1)/2$ rounds in the \congest{B} model with $B < \log((\frac{\widetilde{n}-2}{2e(\Delta-1)})^{\Delta-1})$. We show an algorithm $A$ that solves the fixed-node triangle membership problem on $(G=(V,E),u)$ in a single round of the \congest{B} model with $B < \log((\frac{\widetilde{n}-2}{2e(\Delta-1)})^{\Delta-1})$. As this contradicts Theorem~\ref{thm: FLTD}, and $n=k\widetilde{n}$, this completes the proof for all values of $k$ in $O(n^{1-\epsilon})$ for any constant $0\leq\epsilon\leq 1$.
	
	We construct $A$ such that nodes in $G$ simulate the nodes in $G'$ running algorithm $A'$, as follows. In $A$, each node in $V$, sends to its neighbors in $G$ the message it sends in the first round of $A'$ on $G'$. If $v \in V$ is a neighbor of $u$ in $G$, then it sends to $u$ the message it sends to $uv_{(k-3)/2}$ in $A'$.
	Then, for each of its neighbors $v$, the node $u$ has all the messages that the node $uv_{(k-3)/2}$ receives in the first round of $A'$. Now, by backwards induction, for each $i$, $1\leq i \leq (k-3)/2$, by local simulation at the node $u$, it knows the view of the node $uv_{i}$ at the end of the first $(k-3)/2-i+1$ rounds of $A'$. This implies that $u$ knows the message sent to it by $uv_{1}$ in round $(k-3)/2+1=(k-1)/2$ of $A'$. Since in $(k-1)/2$ rounds of $A'$ the node $u$ knows whether it is in a $k$-cycle in $G'$, by Observation~\ref{ob: triangle-4cycle reduction}, it thus knows in a single round whether it is in a triangle in $G$.
	
To handle even value of $k$, we construct $G'$ in a similar manner of replacing each edge $(u,v) \in E(u)$ by a path of $(k-4)/2$ nodes. In addition, each edge $(v,w) \in E\setminus E(u)$, is replaced by a path of length two which consists of an additional node $vw$. Similarly to Observation~\ref{ob: triangle-4cycle reduction}, the node $u$ participates in a $k$-cycle in $G'$ if and only if it participates in a triangle in $G$. As in the case for odd $k$, given an algorithm $A'$ for $k$-cycle membership in $G'$, the simulation of $(k-4)/2+1=(k-2)/2=\floor{(k-1)/2}$ rounds of it in $G$ requires only a single round of communication. Therefore, the reduction carries over for even values of $k$ as well. Observe that here $n=k\widetilde{n}+\widetilde{n}^2$, therefore, as in the case of odd $k$, we achieve the same asymptotic lower bound as in the triangle membership problem, for any value of $k$ in $O(n^{1-\epsilon})$.
\end{proofof}

\subsection{\texorpdfstring{$k$}{k}-Cycle membership for \texorpdfstring{$k \geq 4$}{k>3} using communication complexity}
\label{subsec:kCycleCC}

Here we show a lower bound on the bandwidth needed for any optimal-round algorithm for solving $k$-cycle membership for $k \geq 4$ by using the standard framework of reduction from a 2-party communication complexity problem. For simplicity, we will show this for even values of $k$, but a similar construction works for odd values as well. For $k>4$, the bound is larger compared with our proof of Section~\ref{sec:kcycle-reduction}.

\begin{theorem}\label{thm: k-cycle}
	The $k$-cycle membership problem cannot be solved by a deterministic optimal-round algorithm in the \congest{B} model unless $B \geq \Omega(\Delta^{\frac{k-4}{2}+1}\log(n))$, and it cannot be solved by a randomized optimal-round algorithm, which succeeds with high probability\footnote{We say that an event occurs with high probability if it occurs with probability $1-\frac{1}{n^c}$, for some constant $c\geq 1$.}, in the \congest{B} model unless $B \geq \Omega(\Delta^{\frac{k-4}{2}+1})$, for any integers $n,\Delta,k$ such that
		$$(\Delta-1)^{\frac{k-4}{2}+1}=O(n^{1-\epsilon}),$$
	for some constant $0<\epsilon<1.$
\end{theorem}

In order to prove Theorem~\ref{thm: k-cycle}, we show a reduction from the $2$-party communication complexity problem $DISJ_{K,S}$.

A 2-party communication complexity problem~\cite{KushilevitzNBook} consists of a function $f:\set{0,1}^K\times\set{0,1}^K\to\set{\true,\false}$, and two strings, $x,y\in\set{0,1}^K$, that are given as inputs for two players, Alice and Bob, respectively. The players exchange bits of communication in order to compute $f(x,y)$, according to a protocol $\pi$. The \emph{communication complexity} $CC(\pi)$ of a protocol $\pi$ for computing $f$ is the maximal number of bits, taken over all input pairs $(x,y)$, exchanged between Alice and Bob. The \emph{communication complexity} $CC(f)$ of $f$ is the minimum, taken over all  protocols $\pi$ that compute $f$, of $CC(\pi)$.

In the \emph{S-Disjointness} problem ($DISJ_{K,S}$), each of the players Alice and Bob receives a $K$-bits input string containing exactly $S$ ones, and the function $f$ is $DISJ_{K,S}(x,y)$, whose output is $\false$ if there is an index $i\in \{0,...,K-1\}$ such that $x_i=y_i=1$, and $\true$ otherwise. Observe that for $S>\frac{K}{2}$, the function is constant. For $S\leq \frac{K}{2}$ The deterministic communication complexity of $DISJ_{K,S}$ is known to be $\Omega(\log \binom{K}{S})$\cite{Jukna,KushilevitzNBook}\footnote{It can be proved by the rank method for proving lower bounds for deterministic protocols in communication complexity. To read more about the rank method, see for example~\cite{KushilevitzNBook} section 1.4. For the proof of the lower bound on the rank of $S$-Disjointness, see~\cite{Jukna}, page 175.}, while its randomized communication complexity is known to be $\Omega(S)$\cite{Razborov92,HastadW07}.
For the reduction, we adapt the formalization of \emph{Family of Lower Bound Graphs} given in~\cite{Censor-HillelKP17} to our setting.

\begin{definition}[\textbf{Definition 1 (simplified) of~\cite{Censor-HillelKP17}: Family of Lower Bound Graphs}]
	\label{def:family}
	Fix an integer $K$, a function $f:\set{0,1}^K\times\set{0,1}^K\to\set{\true,\false}$. The family of graphs $\{G_{x,y}=(V,E_{x,y})\mid x,y\in\set{0,1}^K\}$, is said to be a family of \emph{lower bound graphs w.r.t. $f$ and k-cycle membership} if the following properties hold:
	\begin{enumerate}
		\item[(1)] The set of nodes $V$ is the same for all graphs, and we denote by $V=\{u\}\dot\cup V_A\dot\cup V_B\dot\cup \widetilde{V}$ a fixed partition of it;
		\item[(2)] Only the existence of edges in $V_A\times \widetilde{V}$ may depend on $x$;
		\item[(3)] Only the existence of edges in $V_B\times \widetilde{V}$ may depend on $y$;
		\item[(4)] The node $u$ participates in a $k$-cycle in $G_{x,y}$ iff $f(x,y)=\false$.
	\end{enumerate}
\end{definition}

Observe that given a family of lower bound graphs $\{G_{x,y}=(V,E_{x,y})\mid x,y\in\set{0,1}^K\}$ w.r.t. to $DISJ_{K,S}$ and $k$-cycle membership, if Alice and Bob can simulate an algorithm for $k$-cycle membership on $u$, then by checking the output of $u$ at the end of the algorithm they can solve $DISJ_{K,S}(x,y)$.

The proof of Theorem~\ref{thm: k-cycle} is organized as follows. First, we construct a family of lower bound graphs, and next, we show that given an algorithm $ALG$ for $k$-cycle membership with messages of size $B$, Alice and Bob can simulate $ALG$ on $G_{x,y}$ by exchanging only $O(B)$ bits.

We now construct the following family of lower bound graphs, by describing a fixed graph construction $G=(V,E)$, which we then generalize to a family of graphs $\{G_{x,y}=(V,E_{x,y})\mid x,y\in\set{0,1}^K\}$, which we show to be a family lower bound graphs w.r.t. to $DISJ_{K,S}$ and $k$-cycle membership.

~\\
\noindent\textbf{The fixed graph construction:}
The fixed graph construction (Figure~\ref{fig: CCreduction}) consists of a tree $T$ and a path $P$ of size $n-|T|$. The tree $T$ is a tree in which the root node of $T$, denoted by $u$, is connected to two nodes $w$ and $v$, such that each of $w$ and $v$ is a root of a $(\Delta-1)$-regular tree of depth $\frac{k-4}{2}$. Denote by $leaves(w)$ and $leaves(v)$ the set of leaves of the tree rooted at $w$ and the set of leaves of the tree rooted at $v$, respectively. We define $V_A$ and $V_B$ to be $leaves(w)$ and $leaves(v)$ respectively.

\begin{figure}
	\centering
	\includegraphics[scale=0.6, trim={5cm 6.5cm 5cm 7cm}, clip]{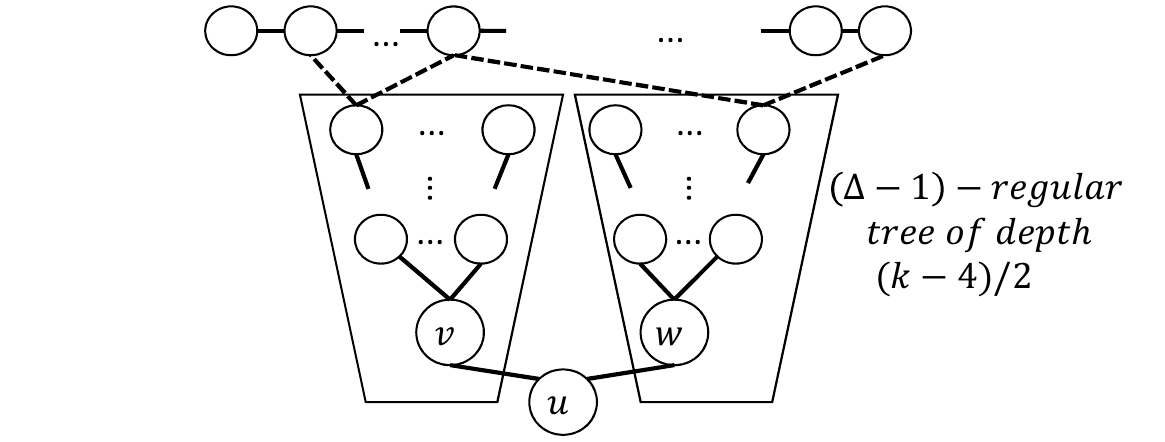}
	\caption{The reduction from Communication Complexity to $k$-cycle membership.}
	\label{fig: CCreduction}
\end{figure}

~\\
\noindent\textbf{Adding edges corresponding to the inputs:} Each of the players Alice and Bob receives as the input a set of nodes in $P$ of size $(\Delta-1)^{\frac{k-4}{2}+1}$. Alice connects the nodes in her input to the $(\Delta-1)^{\frac{k-4}{2}}$ leaves of the tree rooted at $v$, such that each leaf is connected to $\Delta-1$ nodes in $P$. Similarly, Bob connects the nodes in his input to the $(\Delta-1)^{\frac{k-4}{2}}$ leaves of the tree rooted at $w$, such that each leaf in connected to $\Delta-1$ nodes in $P$. The following observation follows directly from the construction.

\begin{observation}\label{ob: k-cycle}
	The node $u$ participates in a cycle of length $k$ in $G_{x,y}$ if and only if the two sets of Alice and Bob are not disjoint.
\end{observation}

Therefore, given an algorithm $ALG$ for $k$-cycle membership, if Alice and Bob can simulate $ALG$ on $u$ then they can solve $S$-Disjointness, where $S=(\Delta-1)^{\frac{k-4}{2}+1}$ and the size of the input stings is $$K=n-|T|\geq n-\left(1+2\sum_{i=0}^{\frac{k-4}{2}}(\Delta-1)^i\right)\geq n-\left(1+2\frac{(\Delta-1)^{\frac{k-4}{2}+1}-1}{\Delta-2}\right)$$
Observe that for
$$(\Delta-1)^{\frac{k-4}{2}+1}=O(n^{1-\epsilon})$$ for some constant $0<\epsilon< 1$, it holds that $K=\Theta(n)$, and $\log \binom{K}{S}=\Omega(\Delta^{\frac{k-4}{2}+1}\log(n))$. It remains to show that given an algorithm for $k$-cycle membership with messages of size $B$, Alice and Bob can simulate $ALG$ on $u$ by exchanging only $O(B)$ bits.
\begin{proofof}{Theorem~\ref{thm: k-cycle}}
	Let $ALG$ be a ($k/2-1$)-round algorithm for solving the $k$-cycle membership problem. Let $\{m^{1}_{v\rightarrow u},...,m^{k/2-1}_{v\rightarrow u}\}$ and $\{m^1_{w\rightarrow u},...,m^{k/2-1}_{w\rightarrow u}\}$ be the two sets of messages sent from $v$ and $w$ to $u$ during the $k/2-1$ rounds, where, e.g., $m^i_{v\rightarrow u}$  is the message sent from $v$ to $u$ in round $i$.
	
	The crucial observation is that Alice and Bob can simulate the nodes $v$, $w$, and $u$ during the first $k/2-2$ rounds without any communication, because the $k/2-2$ neighborhoods of these nodes are fixed. Therefore, in order for the players to compute the output of $u$ after $k/2-1$ rounds, it suffices for Alice to send to Bob the message $m^{k/2-1}_{v\rightarrow u}$, and for Bob to send to Alice the message $m^{k/2-1}_{w\rightarrow u}$.
	
	By Observation~\ref{ob: k-cycle}, this implies that $2B$ bits suffice for the players to correctly compute $DISJ_{K,S}$. Therefore, by the lower bounds on $DISJ_{K,S}$, any deterministic optimal-round algorithm for $k$-cycle membership requires messages of size $\Omega(\log\binom{K}{S})=\Omega(\Delta^{\frac{k-4}{2}+1}\log(n))$, and any randomized optimal-round algorithm which succeeds with high probability requires $\Omega(S)=\Omega(\Delta^{\frac{k-4}{2}+1})$, for any integers $n,\Delta,k$ such that $$(\Delta-1)^{\frac{k-4}{2}+1}=O(n^{1-\epsilon})$$ for some constant $0<\epsilon< 1$.
	
\end{proofof}

We mention that one can use the set-disjointness function instead of $DISJ_{K,S}$ to obtain a lower bound of $(\Delta-1)^{\frac{k-4}{2}}$, by simply connecting each leaf to a single node on the path, based on the inputs. However, this gives a rather strong bound on $\Delta$ with respect to $n$ because then the number of leaves must be also linear in $n$. Such a bound for $\Delta$ does not occur in our given construction.

\end{document}